\documentclass[runningheads, envcountsame, a4paper,11pt]{llncs}

\usepackage{fullpage}
\usepackage{url}
\usepackage{xspace}
\usepackage[utf8]{inputenc}
\usepackage{hyperref}
\usepackage{tikz}
\usetikzlibrary{arrows,automata}

\newcommand{\first}{{\rm first}}
\newcommand{\last}{{\rm last}}
\newcommand{\bLSP}{{\rm bLSP}}
\newcommand{\set}{{\rm set}}
\newcommand{\pref}{{\rm pref}}
\newcommand{\alp}{{\rm alph}}

\newcommand{\SbLSP}{S_{\rm bLSP}}
\newcommand{\bw}{\ensuremath{\mathbf{w}}\xspace}
\newcommand{\bF}{\ensuremath{\mathbf{F}}\xspace}
\newcommand{\mbf}{\ensuremath{\mathbf{f}}\xspace}

\begin{document}
\title{A Characterization of Infinite LSP Words\footnote{This paper, without the proof 
of Proposition~\ref{P:carac1_binary}, the proof of Lemma~\ref{L:finite} and Lemma~\ref{L:extendable}, has been accepted for publication in the proceedings of conference \href{http://www.cant.ulg.ac.be/dlt}{Developments in Language Theory} (DLT 2017). The final publication will be available at \url{link.springer.com}. Many thanks to referees for their careful readings and their interesting suggestions and questions.}}
\toctitle{A Characterization of Infinite LSP Words}
\author{Gwena\"el Richomme}
\tocauthor{Gwena\"el~Richomme}
\institute{
Univ. Paul-Val\'{e}ry Montpellier 3, UFR 6, Dpt MIAp, Case J11,\\
Rte de Mende, 34199 Montpellier Cedex 5, France\\ and \\
LIRMM (CNRS, Univ. Montpellier), UMR 5506 - CC 477,\\ 
161 rue Ada,  34095 Montpellier Cedex 5, France\\
\email{gwenael.richomme@lirmm.fr}
}
\authorrunning{G. Richomme}

\maketitle

\begin{abstract}
G. Fici proved that a finite word has a minimal suffix automaton 
if and only if 
all its left special factors occur as prefixes. 
He called LSP all finite and infinite words having this latter property. 
We characterize here infinite LSP words in terms of $S$-adicity. More precisely we provide a finite set of morphisms $S$ and an automaton ${\cal A}$ such that an infinite word is LSP if and only if it is $S$-adic and all its directive words are recognizable by ${\cal A}$.
\end{abstract}

\section{Introduction}

Extending an initial work by M.~Sciortino and L.Q.~Zamboni \cite{Sciortino_Zamboni2007DLT},
G.~Fici investigated relations between the structure of the suffix automaton built from a finite word $w$ and the combinatorics of this word \cite{Fici2011TCS}. 
He proved that words having their associated automaton with a minimal number of states (with respect to the length of $w$) are the words having all their left special factors as prefixes. 
G.~Fici asked in the conclusion of his paper for a characterization of the set of words having the previous property, that he called the LSP property, both in the finite and the infinite case. 
We provide such a characterization for infinite words in the context of $S$-adicity.

We assume that readers are familiar with combinatorics on words; for omitted definitions (as for instance, factor, prefix, ...) see, \textit{e.g.}, \cite{Lothaire1983book,Lothaire2002,Berthe_Rigo2010CANT}. 
Given an alphabet $A$, $A^*$ is the set of all finite words over $A$ and $A^\omega$ is the set of all infinite words over $A$. A finite word $u$ is a \textit{left special factor} of a finite or infinite word $w$ if there exist at least two distinct letters $a$ and $b$ such that both $au$ and $bu$ occur in $w$. Given two alphabets $A$ and $B$, a \textit{morphism} (\textit{endomorphism} when $A = B$) $f$ is a map from $A^*$ to $B^*$ such that for all words $u$ and $v$ over $A$, $f(uv) = f(u)f(v)$. Morphisms extend naturally to infinite words.

Let $S$ be a set of morphisms. An infinite word $\bw$ is said $S$-adic if there exists a sequence $(f_n)_{n\geq 1}$ of morphisms in $S$
and a sequence of letters $(a_n)_{n\geq 1}$ such that 
$\lim_{n \to +\infty} |f_1f_2\cdots f_n(a_{n+1})|$ $= +\infty$ and
$\bw = \lim_{n \to +\infty} f_1f_2\cdots f_n(a_{n+1}^\omega)$.
The sequence $(f_n)_{n\geq 1}$ is called the \textit{directive word} of $\bw$. 
We consider here $S$-adicity in a rather larger way:
a word $\bw$ is \textit{$S$-adic} with directive word $(f_n)_{n\geq 1}$ if there exists an infinite sequence of infinite words $(\bw_n)_{n \geq 1}$ such that $\bw_1 = \bw$ and $\bw_n = f_n( \bw_{n+1})$ for all $n \geq 1$. Denoting $w_k = f_kf_{k+1}\cdots f_n(a_{n+1}^\omega)$ shows that if the former definition is verified, the latter is also verified. 
This second definition may include degenerated cases as, for instance, the word $a^\omega$ that is $\{Id\}$-adic with $Id$ the morphism mapping $a$ on $a$. For more information on $S$-adic systems, readers can consult, \textit{e.g.}, papers \cite{Berthe2016RIMS,Berthe_Delecroix2014RIMS} and their references.

Let $p_\bw$ be the \textit{factor complexity} of the infinite word $\bw$, that is the function that counts the number of different factors of $\bw$. If $\bw$ is an infinite LSP word, by definition, it has at most one left special factor of each length. Thus it is well-known that $p_\bw(n+1)-p_\bw(n) \leq \#A-1$ (where for any set $X$, $\#X$ denotes the cardinality of $X$).  We let readers verify that all infinite LSP words are uniformly recurrent (all factors occur infinitely many times with bounded gaps). 
By a result of S. Ferenczi \cite{Ferenczi1996ETDS} (see also \cite{Leroy2012thesis,Leroy2014DMTCS,Leroy_Richomme2012integers}), 
there exists a finite set $S$ of morphisms such that all infinite LSP words are $S$-adic. But this general result does not provide a characterization of infinite LSP words.

Our characterization is twofold. First we exhibit an adapted finite set of morphisms $\SbLSP$. Second we show that there exists an automaton that recognizes the set of directive words of infinite LSP words.
In the binary case, our result can be seen as a version for infinite words
of a result of M.~Sciortino and L.Q.~Zamboni \cite{Sciortino_Zamboni2007DLT} (see the conclusion).
In the ternary case, morphisms in $\SbLSP$ are the mirror morphisms of Arnoux-Rauzy-Poincar\'e morphisms (here $f$ is a \textit{mirror morphism} of $g$ if $f(a)$ is the mirror image or reversal of $g(a)$ for all letters $a$).
These morphisms were used by V.~Berth\'e and S.~Labb\'e \cite{Berthe_Labbe2015AAM} to provide an S-adic system recognizing sequences arising from the study of the Arnoux-Rauzy-Poincar\'e multidimensional continued fraction algorithm. 
For alphabets of cardinality at least $4$, new morphisms appear.

The paper is organized as follows.
After introducing in Section~\ref{sec:morphisms} our basis of morphisms $\SbLSP$,
in Section~\ref{sec:s-adicity}, we show that all infinite LSP words are $\SbLSP$.
Section~\ref{sec:fragility} introduces a property of infinite LSP words and a property of morphisms in $\SbLSP$ that allow to explain why the LSP property is lost when applying a LSP morphism to an infinite LSP word.
Section~\ref{sec:origin} allows to trace the origin of the previous property of infinite LSP words.
Based on this information, Section~\ref{sec:automaton} defines our automaton and
Section~\ref{sec:carac} proves our characterization of infinite LSP words.
We end with a few words on characterizations of finite LSP words.

\section{\label{sec:morphisms}Some Basic Morphisms}

We call \textit{basic LSP morphism} on an alphabet $A$, or $\textit{bLSP}$ in short, 
any endomorphism $f$ of $A^*$ that verifies:
\begin{itemize}
\itemsep0cm
\item there exists a letter $\alpha$ such that $f(\alpha) = \alpha$, and
\item for all letters $\beta \neq \alpha$, there exists a letter $\gamma$ such that $f(\beta) = f(\gamma) \beta$
\end{itemize}
We let $\SbLSP(A)$ (or shortly $\SbLSP$ when $A$ is clear) denote the set of all bLSP morphisms over the alphabet $A$. Observe that for any bLSP morphism $f$, there exists a unique letter $\alpha$ such that $f(\alpha) = \alpha$. We let $\first(f)$ denote this letter as it is also the first letter of $f(\beta)$ for any letter $\beta$. 
We also let $[ u_1, u_2, \ldots]$ denote the morphism defined by $a \mapsto u_1$, $b \mapsto u_2$, \ldots. 
For instance, $[a, ab, abc, abcd, abcde]$  defines the morphism $f$ such that $f(a) = a$, $f(b) = ab$, $f(c) = abc$, $f(d) = abcd$, $f(e) = abcde$.

\begin{remark}\rm
By definition of bLSP morphisms, given an alphabet $A$, there is a bijection between $\SbLSP(A)$ 
and the set of labeled rooted trees with label in $A$ (all labels are on vertices and distinct vertices have distinct labels). Given a labeled rooted tree $T = (A, E)$, the associated bLSP morphism $f$ is the one such that, for all letters $\beta$, $f(\beta)$ is the word obtained concatenating vertices on the path in $T$ from the root to $\beta$.
For instance, the rooted trees associated with morphisms  $[a, ab, abc, abcd]$,
 $[a, ab, abc, abd]$,
 $[a, ab, abc, ad]$ and 
 $[a, ab, ac, ad]$ are given in Figure~\ref{trees}.

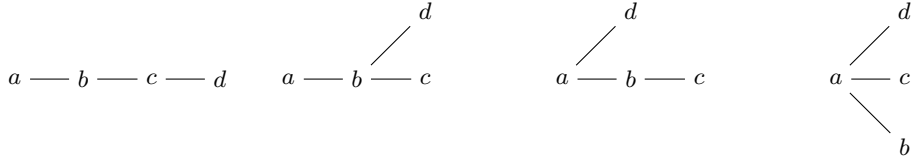
\begin{figure}[!h]
\begin{center}
\begin{tikzpicture}[scale=0.3]
\begin{scope}
\node (a) at (0,0) {$a$};
\node (b) at (3,0) {$b$};
\node (c) at (6,0) {$c$};
\node (d) at (9,0) {$d$};
\draw (a) -- (b);
\draw (b) -- (c);
\draw (c) -- (d);
\end{scope}

\begin{scope}[xshift=12cm]
\node (a) at (0,0) {$a$};
\node (b) at (3,0) {$b$};
\node (c) at (6,0) {$c$};
\node (d) at (6,3) {$d$};
\draw (a) -- (b);
\draw (b) -- (c);
\draw (b) -- (d);
\end{scope}

\begin{scope}[xshift=24cm]
\node (a) at (0,0) {$a$};
\node (b) at (3,0) {$b$};
\node (c) at (6,0) {$c$};
\node (d) at (3,3) {$d$};
\draw (a) -- (b);
\draw (b) -- (c);
\draw (a) -- (d);
\end{scope}

\begin{scope}[xshift=36cm]
\node (a) at (0,0) {$a$};
\node (b) at (3,-3) {$b$};
\node (c) at (3,0) {$c$};
\node (d) at (3,3) {$d$};
\draw (a) -- (b);
\draw (a) -- (c);
\draw (a) -- (d);
\end{scope}
\end{tikzpicture}
\end{center}
\caption{\label{trees}Rooted trees associated with bLSP morphisms}
\end{figure}

\end{remark}
The previous remark allows to enumerate bLSP morphisms (see Sequence A000169 in \href{https://oeis.org/A000169}{The On-}\\\href{https://oeis.org/A000169}{Line Encyclopedia of Integer Sequences} whose first values are $1$, $2$, $9$, $64$, $625$, $7776$, $117649$, $2097152$).

Here follows some examples of bLSP morphisms. 
\begin{itemize} 
\item $\SbLSP(\{a, b\}) = \{ [a, ab], [ba, b] \}$. These morphisms are well-known in the context of Sturmian words. They are denoted $\tau_a$ and $\tau_b$ in \cite{BertheHoltonZamboni2006} from which it can be seen that standard Sturmian words are non-periodic $\{\tau_a, \tau_b\}$-adic words (see also \cite{LeveRichomme2007TCS}). 

\item $\SbLSP(\{a, b, c\}) = \{ 
[a, ab, abc],$ 
$[a, ab, ac],$
$[a, acb, ac],$
$[ba, b, bac],$ 
$[ba, b, bc],$ 
$[bca, b, bc],$
$[ca, cb, c],$ 
$[ca, cab, c],$ 
$[cba, ca, c]\} =$
$\{ p^{-1} \circ [a, ab, abc] \circ p, p^{-1} \circ [a, ab, ac] \circ p \mid p \in perm(A)\}$
where $perm(A)$ is the set of all endomorphisms of $A^*$ whose restriction to the set of letters is a permutation of the alphabet.
As mentioned in the introduction, these sets $\SbLSP(\{a, b, c\})$ is also the set of mirror morphisms considered in \cite{Berthe_Labbe2015AAM}, that is mirrors of the Poincaré substitutions (defined for $\{i, j, k\} = \{a, b, c\}$ by $i \mapsto ijk$, $j \mapsto jk$, $k \mapsto k$) and the Arnoux-Rauzy substitutions (defined for $\{i, j, k\} = \{a, b, c\}$ by $i \mapsto ik$, $j \mapsto jk$, $k \mapsto k$).

\item The set $\SbLSP(\{a, b, c, d\})$ is the set of all morphisms on the form $p^{-1}\circ f \circ p$ for $p \in perm(A)$, and $f$ being one of the following morphisms:
 $[a, ab, abc, abcd]$,
 $[a, ab, abc, abd]$,
 $[a, ab, abc, ad]$ and 
 $[a, ab, ac, ad]$.

\end{itemize}

We end this section with some basic properties of bLSP morphisms that follow directly from the definition.
For a non-empty word $u$, let $\first(u)$ denote its first letter, $\last(u)$ its last letter and $\alp(u)$ its set of letters.

\begin{property}
\label{bLSP properties}Let $f$ be a bLSP morphism over the alphabet $A$.
\begin{enumerate}
\item there exists a unique letter $\alpha \in A$ such that for all $\beta \in A$, $\first(f(\beta)) = \alpha$;
\item for all $\beta \in A$, $\last(f(\beta)) = \beta$;
\item \label{rem:unique letter} there exists a unique letter $\alpha \in A$ such that $f(\alpha) = \alpha$: $\alpha = \first(f)$;
\item $f(A)$ is a suffix code (no word of $f(A)$ is a suffix of another word in $f(A)$);
\item $f$ is injective both on the set of finite words and the set of infinite words;
\item for all $\beta \in A$, $x$, $y \in A^*$, if $|x| = |y|$ and if $x\beta$ and $y\beta$ are  factors of words in $f(A)$, then $x = y$;
\item \label{property1} for all letters $\beta$, $\gamma$, $|f(\beta)|_\gamma \leq 1$.
\end{enumerate}

\end{property}

\section{\label{sec:s-adicity}$\SbLSP$-Adicity of Infinite LSP Words}

\begin{proposition}
\label{characLSPwords} Any infinite LSP word is $\SbLSP$-adic.
\end{proposition}

Given a set $S$ of morphisms, in order to prove that infinite words verifying a property $P$ are $S$-adic, it suffices to prove that for all infinite words $\bw$ verifying $P$ that:
\begin{enumerate}
\item there exists $f \in S$ and an infinite word $\bw'$ such that $\bw = f(\bw')$, and
\item if $\bw = f(\bw')$ with $f \in S$, then $\bw'$ verifies Property $P$.
\end{enumerate}

Hence Proposition~\ref{characLSPwords} is a direct consequence of the next two lemmas.

\begin{lemma}
\label{lemma2}
Given any finite or infinite LSP word $\bw$, 
there exist a bLSP morphism $f$ on $alph(\bw)$ and an infinite word $\bw'$ such that $\bw = f(\bw')$.
\end{lemma}

\begin{proof}
Let $\bw$ be a non-empty finite or infinite LSP word and let $\alpha$ be its first letter.
Let $X$ be the set of words over $\alp(w)\setminus\{\alpha\}$ such that $w$ can be factorized over $\{\alpha\} \cup X$.
Let $G$ be the graph $(\alp(w), E)$ with $E$ the set of edges $(\beta, \gamma)$ such that $\beta\gamma$ is a factor of a word $\alpha u$ with $u \in X$. By LSP Property of $\bw$, each letter occurring in a word of $X$ is not left special in $\bw$.
Hence $G$ is a rooted tree with $\alpha$ as root, that is, for any letter $\beta$, there exists a unique path from $\alpha$ to $\beta$. We let $u_\beta$ denote the word obtained by concatenating the letters occurring in the path. Let $f$ be the morphism defined by $f(\beta) = u_\beta$. By construction, $f$ is bLSP and $\bw = f(\bw')$ for a word $\bw'$.
\qed\end{proof}

\begin{remark}\rm
\label{rm_lemma2}
The word $\bw'$ in Lemma~\ref{lemma2} is unique.
The morphism $f$ is not unique but its restriction to $alph(\bw')$ is.
It can also be observed that this restriction is entirely defined by the first letter of $\bw$ and
the factors of length two of $\bw$.
\end{remark}

\begin{lemma}
\label{lemma1}
For any bLSP morphism $f$ and any infinite word $\bw$, if $f(\bw)$ is LSP then $\bw$ is LSP.
\end{lemma}

\begin{proof}
Assume by contradiction that $\bw$ is not LSP. 
This means that $\bw$ has (at least) one left special factor that is not one of its prefixes.
Considering such a factor of minimal length,
there exist a word $u$ and letters $a$, $b$, $\beta$, $\gamma$ such that
$a \neq b$, $\beta \neq \gamma$,
$ua$ is a prefix of $\bw$, 
$\beta ub$ and $\gamma ub$ are factors of $\bw$.
Recall that $f$ is a bLSP morphism: let $\alpha =\first(f)$.
The word $f(u)f(a)\alpha$ is a prefix of $f(\bw)$.
Moreover by Property~\ref{bLSP properties}(2), 
words $\beta f(u)f(b)\alpha$ and $\gamma f(u)f(b)\alpha$
are factors of $\bw$ (here the fact that $\bw$ is infinite is needed: each factor is followed by a letter whose image begins with $\alpha$).
As $f(a) \neq f(b)$ and as the letter $\alpha$
occurs only as a prefix in $f(a)$ and $f(b)$, 
$f(a)\alpha$ is not a prefix of $f(b)\alpha$ and, conversely, 
$f(b)\alpha$ is not a prefix of $f(a)\alpha$.
Hence there exist a word $v$ and letters $\alpha'$, $\beta'$ such that $\alpha' \neq \beta'$, 
$v\alpha'$ and $v\beta'$ are respectively prefixes of $f(a)\alpha$ and $f(b)\alpha$.
It follows that $f(u)v\alpha'$ is a prefix of $f(\bw)$ while
$\beta f(u)v\beta'$ and $\gamma f(u)v\beta'$
are factors of $f(\bw)$: $f(\bw)$ is not LSP.
\qed\end{proof}

Observe that Lemma~\ref{lemma1} does not hold for finite words. For instance the word $baa$ is not LSP while its image $abaa$ by the morphism $[a, ab]$ is LSP. 

To end this section let us mention that in the binary case the converses of Lemma~\ref{lemma1} and Proposition~\ref{characLSPwords} hold.

\begin{proposition}
\label{P:carac1_binary}
If $\bw$ is a binary LSP infinite word and if $f$ 
belongs to the set $\{ [a, ab], [ba, b] \}$ then $f(\bw)$ is also LSP. Consequently a binary word is LSP if and only if it is $\{ [a, ab], [ba, b] \}$-adic.
\end{proposition}


\begin{proof}
The second part of the proposition is a direct consequence of the first part and of Lemma~\ref{lemma2}.
We prove the first part for $f = [a, ab]$ (exchanging the roles of $a$ and $b$, the proof is similar for $f = [ba, b]$). Assume $\bw$ is an  infinite LSP word over $\{a, b\}$ such that $f(\bw)$ is not LSP. There exists a left special factor $v$ of $f(\bw)$ which is not a prefix of $f(\bw)$. Choose $v$ of minimal length. Hence $v = v' \beta$ and $v'\alpha$ is a prefix of $f(\bw)$ for a word $v'$ and two letters $\alpha$ and $\beta$ such that $\{\alpha, \beta \} = \{a, b\}$. As the letter $b$ is always preceded in $f(\bw')$ by the letter $a$, the word $v'$ is a non-empty word ending with $a$: hence $v' = f(u)a$ for some word $u$. 
As $f(u)a\alpha$ is a prefix of $f(\bw)$, by the structure of $f$, $u\alpha$ is a prefix of $\bw$.
Also as $f(u)a\beta$ is a left special factor of $f(\bw)$, as $\last(f(a)) = a$ and $\last(f(b)) = b$, the word $u\beta$ is also a left special factor of $\bw$. As $\alpha \neq \beta$, this contradicts the fact that $\bw$ is LSP.
\qed
\end{proof}

\section{\label{sec:fragility}Fragility of Infinite LSP Words}

For alphabets of cardinality at least $3$, the converse of Lemma~\ref{lemma1} is false: there exist an infinite LSP word $\bw$ and a bLSP morphism $f$ such that $f(\bw)$ is not LSP.
For instance, let $\bF$ be the well-known Fibonacci word (the fixed point of the endomorphism $[ab, a]$), and let $g$ be the bLSP morphism $[a, acb, ac]$. 
The word $g^2(\bF)$ begins with the word $g^2(ab) = g(aacb) = aaacacb$ that contains the factor $ac$ which is left special but not a prefix of the word.
Hence the word $g^2(\bF)$ is not LSP while $\bF$ is LSP and $g$ is bLSP (actually one can prove, using Lemma~\ref{when (a, b, c)-breaking morphisms appear} below, that $g(\bF)$ is LSP). 

In what follows, we introduce some properties of LSP words and morphisms that explain in which context a (breaking) bLSP morphism can map a (fragile) infinite LSP word on a non-LSP word.

\begin{definition}
\label{def_fragility}\rm
Let $a, b, c$ be three pairwise distinct letters. An infinite word $\bw$ is $(a, b, c)$-\textit{fragile} if there exist a word $u$ and distinct letters $\alpha$ and $\beta$ such that the word $ua$ is a prefix of $\bw$ and the words $\alpha u b$ 
and $\beta u c$ are factors of $\bw$. 
We will also say that $\bw$ is $(a, b, c, \alpha, \beta)$-fragile when we need letters $\alpha$ and $\beta$. The word $u$ is also called an $(a, b, c, \beta, \gamma)$-\textit{fragility} of $\bw$.
\end{definition}

For instance, the empty word $\varepsilon$ is an $(a, b, c, c, a)$-fragility of $g(\bF)$: $\varepsilon a$ is a prefix of $g(\bF) = aacb\cdots$ while $c\varepsilon b$ and $a\varepsilon c$ are factors of  $g(\bF)$.
More generally any factor $abc$ or $acb$ in an infinite word starting with the letter $a$ 
(and with $a \neq b \neq c \neq a$)
produces an $(a,b,c)$-fragility. 
One can also observe that, by symmetry of the definition, any $(a, b, c)$-fragile word is also $(a, c, b)$-fragile. Finally let us note that no fragility exists in words over two letters (as the definition needs three pairwise distinct letters).

The main idea of introducing the previous notion is that for any $(a, b, c)$-fragile LSP word $\bw$, there exists a bLSP morphism such that $f(\bw)$ is not LSP.
For instance, if $u, \alpha, \beta, \bw$ are as in Definition~\ref{def_fragility}, the word $g(u)aa$ is a prefix of $g(\bw)$ whereas words $\alpha g(u)acb$ and $\beta g(u)ac$ are factors of $g(\bw)$, so that $g(\bw)$ is not LSP since $g(u)ac$ is left special but not a prefix of $g(\bw)$.

\begin{definition}\rm
Let $a, b, c$ be three pairwise distinct letters.
A morphism $f$ is LSP $(a, b, c)$-\textit{breaking}, if for all $(a, b, c)$-fragile infinite LSP word $\bw$, $f(\bw)$ is not LSP.
\end{definition}

For instance, the morphism $g = [a, acb, ac]$ is $(a, b, c)$-breaking.

\begin{lemma}
\label{when (a, b, c)-breaking morphisms appear}
Let $\bw$ be an infinite LSP word and let $f$ be a bLSP morphism.
The following assertions are equivalent:
\begin{enumerate}
\item The word $f(\bw)$ is not LSP;
\item There exist some pairwise distinct letters $a, b, c$ such that $\bw$ is $(a, b, c)$-fragile and 
the longest common prefix of $f(b)$ and $f(c)$ is  strictly longer than the longest common prefix of $f(a)$ and $f(b)$;
\item There exist some pairwise distinct letters $a, b, c$, such that $\bw$ is $(a, b, c)$-fragile 
and $f$ is LSP $(a, b, c)$-breaking.
\end{enumerate}
\end{lemma}

\begin{proof}
$1 \Rightarrow 2$.
Assume first that $f(\bw)$ is not LSP. 
There exists a left special factor $V$ of $f(\bw)$ which is not a prefix of $f(\bw)$.
Let $v$ be the longest common prefix of $V$ and $f(\bw)$.
Let $a', b'$ be the letters
such that $va'$ is a prefix of $f(\bw)$ 
and $vb'$ is a prefix of $V$: by construction $a' \neq b'$.
Let also $\beta$, $\gamma$ be distinct letters such that $\beta V$ and $\gamma V$ are factors of $f(\bw)$ (also $\beta v b'$ and $\gamma v b'$ are factors of $f(\bw)$).

By Property~\ref{bLSP properties}, 
the letter $\alpha = \first(f)$ is
the unique letter that can be left special in $f(\bw)$.
This implies $v \neq \varepsilon$ and $\first(v) = \first(f)$.
As $\alpha$ occurs exactly at the first position in all images of letters, 
occurrences of $\alpha$ mark the beginning of images of letters in $f(\bw)$.
Considering the last occurrence of $\alpha$ in $v$,
we can write $v = f(u)\alpha x$ with $|x|_\alpha = 0$.
Let $a$, $b$, $c$ be letters such that:
\begin{itemize}
\item $ua$ is a prefix of $\bw$, and, $va' = f(u)\alpha x a'$ is a prefix of $f(ua)$ when $a' \neq \alpha$ or $v = f(ua)$ when $a' = \alpha$;
\item $\beta ub$ is a factor of $\bw$, and, $\beta vb'$ is a prefix of $\beta f(ub)$ when $b' \neq \alpha$ or $v = f(ub)$ when $b' = \alpha$;
\item $\gamma uc$ is a factor of $\bw$, and, $\gamma vb'$ is a prefix of $\gamma f(uc)$  when $b' \neq \alpha$ or $v = f(uc)$ when $b' = \alpha$.
\end{itemize}

As $a' \neq b'$, we have $a \neq b$ and $a \neq c$.
Observe that until now we did not use the fact that $\bw$ is LSP.
This implies $b \neq c$ (and so $b' \neq \alpha$).
Indeed otherwise $ub$ would be a left special factor of $\bw$ without being one of its prefixes: a contradiction with the fact that $\bw$ is an LSP word.
Thus $\bw$ is $(a, b, c)$-fragile.

This ends the proof of Part $1 \Rightarrow 2$ as $\alpha xb'$ is a common prefix of $f(b)$ and $f(c)$ and $\alpha x$ is the longest common prefix of $f(a)$ and $f(b)$.

$2 \Rightarrow 3$.
By hypothesis, $f(a) = v \delta w_1$, $f(b) = v \gamma w_2$ and $f(c) = v \gamma w_3$ for letters $\delta, \gamma$ 
and words $w_1$, $w_2$ and $w_3$ with $\delta \neq \gamma$.
Let $\bw'$ be any LSP $(a, b, c)$-fragile infinite word. 
Let $u'$, $\alpha'$, $\beta'$ be the word and letters such that 
$u'a$ is a prefix of $\bw'$ while $\alpha' u' b$ and $\beta' u' c$ are factors of $\bw'$ with $\alpha' \neq \beta'$.
The word $f(\bw')$ has $f(u')v\delta$ as a prefix and words $\alpha'f(u')v\gamma$ and $\beta'f(u')v \gamma$ as factors.
As $\delta \neq \gamma$, the word $f(\bw')$ is not LSP. The morphism $f$ is LSP $(a,b,c)$-breaking.

$3 \Rightarrow 1$.
This follows the definition of $(a, b, c)$-fragile words and LSP $(a, b, c)$-breaking morphisms.
\qed\end{proof}

Observe that we have also proved the next result.

\begin{corollary}
\label{carac LSP breaking morphisms}
A bLSP morphism is LSP $(a, b, c)$-breaking for pairwise distinct letters $a$, $b$ and $c$ if and only if the longest common prefix of $f(b)$ and $f(c)$ is strictly longer than the longest common prefix of $f(a)$ and $f(b)$.
\end{corollary}

\section{\label{sec:origin}Origin of Fragilities}

Before characterizing infinite LSP words, we need to know how fragilities in an LSP word can appear. 
This is explained by next result. 
For a set $X$ of words, we let ${\rm Fact}(X)$ denote the set of factors of words in $X$.

\begin{lemma}
\label{origine fragilities}
Assume a word $u$ is an $(a, b, c, \beta, \gamma)$-fragility of $f(\bw)$ for an infinite word $\bw$ (not necessarily LSP) over an alphabet $A$ and $f$ is a bLSP morphism (by definition of fragilities, $a$, $b$, $c$, $\beta$, $\gamma$ are letters).

\begin{itemize}
\itemsep0cm

\item (New fragilities) If $u = \varepsilon$, then $a = \first(f)$ and $\beta b$, $\gamma c \in {\rm Fact}(f(\alp(\bw)))$.

\item (Propagated fragilities) If $u \neq \varepsilon$, there exist letters $a'$, $b'$, $c'$ in $\alp(\bw)$ and an $(a', b', c', \beta, \gamma)$-fragility $v$ of $\bw$ such that $|v| < |u|$, $f(v)$ is a proper prefix of $u$ and words $ua$, $\beta u b$, $\gamma u c$ are respectively prefixes of $f(va')\alpha$, $\beta f(vb')\alpha$, $\gamma f(vc')\alpha$ with $\alpha =\first(f)$.

\end{itemize}\end{lemma}

\begin{proof}
\textit{(New fragilities)}
If $u = \varepsilon$, it follows from the definition of an $(a, b, c, \beta, \gamma)$-fragility that $a = \first(\bw)$ and $\beta b$, $\gamma c$ are factors of $f(\bw)$.
Now observe that, still by the same definition, $a \not\in \{b, c\}$.
Thus by definition of bLSP morphisms, $a = \first(f)$ and $\beta b$, $\gamma c$ belong to ${\rm Fact}(f(\alp(\bw)))$.

\medskip

\textit{(Propagated fragilities)}
We assume here that $u$ is not empty.
Let $\alpha = \first(f)$.
Considering the last occurrence of $\alpha$ in $u$, observe that the word $u$ can be decomposed in a unique way as $u = f(v) \alpha x$ with $v$, $x$ words such that $|x|_\alpha = 0$.
As $u$ is an $(a, b, c, \beta, \gamma)$-fragility of $f(\bw)$, there exist words $w_1$, $w_2$ and $w_3$ such that:
\begin{itemize}
\itemsep0cm
\item $|w_1|_\alpha = |w_2|_\alpha = |w_3|_\alpha = 0$;
\item $f(v) \alpha x w_1 \alpha$ is a prefix of $f(\bw)$ and $a = \first(w_1\alpha)$;
\item $\beta f(v)\alpha x w_2 \alpha$ and $\gamma f(v) \alpha x w_3 \alpha$ are factors of $f(\bw)$ with $b = \first(w_2\alpha)$ and $c = \first(w_3 \alpha)$.
\end{itemize}
By definition of a bLSP morphism, there exist letters $a'$, $b'$, $c'$ such that $f(a') = \alpha x w_1$, $f(b') = \alpha x w_2$, $f(c') = \alpha x w_3$. 
These letters $a'$, $b'$, $c'$ are pairwise distinct since letters $a = \first(w_1\alpha)$, $b = \first(w_2\alpha)$ and $c = \first(w_3 \alpha)$ are pairwise distinct.
Moreover $va'$ is a prefix of $\bw$ and words $\beta v b'$ and $\gamma v c'$ are factors of $\bw$ 
(remember that $\alpha$ marks the beginning of letters in $f(\bw)$ as 
$f$ is a bLSP morphism). 
Hence the word $v$ is an $(a', b', c', \beta, \gamma)$-fragility of $\bw$.
Finally let us observe that $|v| \leq |f(v)| < |u|$.
\qed\end{proof}

\section{\label{sec:automaton}An Automaton to Follow Fragilities}

In this section, we introduce an automaton that allows to recognize all directive words of LSP words viewed as $\SbLSP$-adic words. We will prove the converse in next section. 
Observe that transitions of the automaton are defined in order to follow fragilities using Lemma~\ref{origine fragilities}.

\begin{definition}\rm
We let ${\cal A}_{bLSP}$ denote the non-deterministic automaton whose elements are described below.
\begin{itemize}
\item The alphabet of ${\cal A}_{bLSP}$ is the set bLSP of basic LSP morphisms.
\item The set of states $Q$ is the set $2^A \times$bLSP$\times 2^{A^5}$.
Hence a state is the data of a sub-alphabet of $A$, of a bLSP morphism and of a set of 5-tuples 
$(a, b, c, \beta, \gamma)$ of letters whose aim is to represent the set of fragilities of a word. 
For a state $q$, we let $\alp(q)$ denote the sub-alphabet of $A$, by $\bLSP(q)$ the morphism and by $\set(q)$ the set of $5$-tuples.
\item The set of transitions $\Delta$ is the set of triples $(q, f, q')$ such that
\begin{enumerate}
\item $f = \bLSP( q )$;
\item $\alp(q) = \alp(f(\alp(q')))$;
\item if $(a, b, c, \beta, \gamma) \in \set(q')$ then $f$ is not LSP $(a,b,c)$-breaking;
\item $\set(q)$ is the set of all 5-tuples $(a,b,c, \beta, \gamma)$ 
such that $a$, $b$, $c$, $\beta$, $\gamma$ are letters of $\alp(q)$, 
$a \neq b \neq c \neq a$, $\beta \neq \gamma$ and one of the  following two conditions holds:
\begin{enumerate}
\item $a = \first(f)$, $\beta b$, $\gamma c$ in ${\rm Fact}(f(alph(q')))$ and $\beta \neq \gamma$;
\item there exist $a', b', c'$ such that $(a', b', c', \beta, \gamma) \in \set(q')$ and 
a word $x$ such that  
$xa \in \pref(f(a')\alpha)$, 
$xb \in \pref(f(b')\alpha)$ and
$xc \in \pref(f(c')\alpha)$
with $\alpha =\first(f)$.
\end{enumerate}
\end{enumerate}
\item All states are initial.
\end{itemize}
\end{definition}

Figure~\ref{Graph2} shows this automaton when the alphabet is $\{a, b\}$.
In this figure, $\tau_a = [a, ab]$ and $\tau_b = [ba, b]$. 
States $q$ with $set(q) \neq \emptyset$ are not drawn since binary infinite LSP words contain no fragilities.
Moreover states
$(\emptyset, \tau_a, \emptyset)$ and $(\emptyset, \tau_b, \emptyset)$, $(\{a\}, \tau_b, \emptyset)$ and $(\{b\}, \tau_a, \emptyset)$ are not drawn
as there are no transition leaving them. 

\begin{figure}[!h]
\begin{center}
\begin{tikzpicture}[scale=0.3]
\begin{scope}
\node[draw, rectangle, rounded corners = 3pt] (a) at (0,0) {$\{a\}, \tau_a, \emptyset$};
\node[draw, rectangle, rounded corners = 3pt] (b) at (10,0) {$\{a,b\}, \tau_b, \emptyset$};
\node[draw, rectangle, rounded corners = 3pt] (c) at (20,0) {$\{a,b\}, \tau_a, \emptyset$};
\node[draw, rectangle, rounded corners = 3pt] (d) at (30,0) {$\{b\}, \tau_b, \emptyset$};
\draw [<-] (a) -- (b) node[midway,above]{$\tau_b$};
\draw [<-] (b) to[bend right] node[below]{$\tau_a$} (c) ;
\draw [<-] (c) to[bend right] node[above]{$\tau_b$}  (b) ;
\draw [->] (c) -- (d) node[midway,above]{$\tau_a$};
\path  (a) edge [loop above] node[above]{$\tau_a$} (a) ;
\path  (c) edge [loop above] node[above]{$\tau_a$} (c) ;
\path  (d) edge [loop above] node[above]{$\tau_b$} (d) ;
\path  (b) edge [loop above] node[above]{$\tau_b$} (b) ;
\end{scope}
\end{tikzpicture}
\end{center}
\caption{\label{Graph2}${\cal A}_{\bLSP}$ for the binary alphabet}
\end{figure}
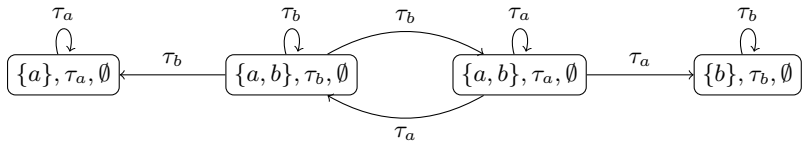

For alphabets with at least three letters, automaton ${\cal A}_{bLSP}$ is too huge to be drawn 
even restricting to states $q$ such that $set(q)$ is a set of fragilities of an LSP word. 

An infinite word $\mbf$ over $\bLSP$ is said to be \textit{recognized} by ${\cal A}_{bLSP}$ 
if there exists an infinite path in ${\cal A}_{bLSP}$ whose label is $\mbf$. The aim of 
${\cal A}_{bLSP}$ is to recognize $\bLSP$ directive words of infinite LSP words.

Let $\bw$ be an LSP word.
We associate with $\bw$ a state of ${\cal A}_{bLSP}$ that we let denote $q(\bw)$. 
This state is the state $q$ such that:
\begin{itemize}
\item $\alp(q) = \alp(\bw)$:
\item $\bLSP(q)$ is any morphism $f$ such that $\bw = f(\bw')$ for some word $\bw'$ (such a morphism exists by Lemma~\ref{lemma2}).
\item $\set(q)$ is the set of all $5$-tuples $(a, b, c, \beta, \gamma)$ such that $\bw$ is $(a, b, c, \beta, \gamma)$-fragile.
\end{itemize}

The fact that, for any LSP word $\bw$, any of its directive word is recognized by $\bw$, is a direct consequence of next lemma.

\begin{lemma}
\label{LSP words are recognizable}
Let $\bw$, $\bw'$ be LSP words such that $\bw = f(\bw')$ with $f = \bLSP(q(\bw))$.
The transition $(q(\bw), \bLSP(q(\bw)), q(\bw'))$ exists in ${\cal A}_{\bLSP}$.
\end{lemma}

\begin{proof}
Let $f = \bLSP(q(\bw))$.
Observe that $\alp(q(\bw)) = \alp(\bw)$, $\bw = f(\bw')$ and $\alp(q(\bw')) = \alp(\bw')$. Whence we have $\alp(q(\bw)) = \alp(f(\alp(q(\bw')))$.

By Lemma~\ref{origine fragilities} and the definition of $q(\bw)$, a 5-tuple $(a, b, c, \beta, \gamma)$ belongs to set $q(\bw)$ if and only if one of the following two conditions holds:
\begin{itemize}
\item $a = \first(f)$, $\beta b$, $\gamma c$ belong to ${\rm Fact}(f(\alp(\bw')))$
\item there exist $a'$, $b'$, $c'$, $\beta$, $\gamma$ in $\alp(\bw')$ and an $(a, b, c, \beta, \gamma)$-fragility $u$ of $\bw$ and an $(a', b', c', \beta, \gamma)$-fragility $v$ of $\bw'$ such that $|v| < |u|$, $f(v)$ is a proper prefix of $u$ and words $ua$, $\beta ub$, $\gamma u c$ are respectively prefixes of $f(va')\alpha$, $\beta f(vb')\alpha)$, $\gamma f(vc')\alpha$ with $\alpha = \first(f)$.
\end{itemize}
For the second case, $u = f(v) x$ for a word $x$. The word $xa$ is a prefix of $f(a')\alpha$,
$xb$ is a prefix of $f(b')\alpha$ and $xc$ is a prefix of $f(c')\alpha$.
As $\bw'$ is  $(a', b', c', \beta, \gamma)$-fragile,  $(a', b', c', \beta, \gamma) \in \set(q(\bw'))$.
Thus in both cases, Condition~4 for $(q(\bw),$ $\bLSP(q(\bw),$ $q(\bw'))$ to be a transition of ${\cal A}_{\bLSP}$ is verified.

To end the proof we have to check Property 3 of transitions of ${\cal A}_{\bLSP}$.
Assume there exists an $(a, b, c, \beta, \gamma)$-fragility in $\set(q(\bw'))$.
By definition of $q(\bw')$, 
this implies that $\bw'$ has an $(a, b, c, \beta, \gamma)$-fragility.
As $\bw = f(\bw')$ is LSP, $f = \bLSP(q(\bw))$ is not LSP $(a, b, c)$-breaking.
\qed\end{proof}

\section{\label{sec:carac}A Characterization of LSP Words}

\begin{theorem}
\label{th_carac}
A word $\bw$ is LSP if and only if it is $S_\bLSP$-adic and all of its directive word are recognized by the automaton ${\cal A}_\bLSP$.
\end{theorem}

\begin{proof}
Proposition~\ref{characLSPwords} and Lemma~\ref{LSP words are recognizable} prove the only if part of Theorem~\ref{th_carac}. Let us prove the if part of Theorem~\ref{th_carac}.

Assume, by contradiction, that ${\cal A}_\bLSP$ recognizes a directive word $\mbf$ of a word $\bw$
which is $S_\bLSP$-adic but not LSP.  
Such a word contains a left special factor $u$ that is not a prefix of $\bw$.
Among all possible triples $(\mbf, \bw, u)$, choose one such that $|u|$ is minimal.

For $n \geq 1$, we let $f_n$ denote the $n^{th}$ letter of $\mbf$ and $\bw_n$ the word directed by $(f_k)_{k\geq n}$ ($\bw_1 = \bw$; $\bw_2$ is directed by 
$f_2 f_3 \cdots$; $\bw_n = f_n( \bw_{n+1})$ for $n \geq 1$).

\medskip

\textit{Step 1: $\bw_2$ contains a fragility} 

First observe $|u| \geq 2$. 
Indeed we have $|u|\neq 0$ as the empty word is a prefix of $\bw$.
Moreover, by the structure of images of the $\bLSP$ morphism $f_1$,
only the letter $\first(f_1)$ can be left special, whence $|u| \neq 1$.

Let $\alpha = \first(\bw) = \first(f_1)$.
Considering the last occurrence of $\alpha$ in $u$,
the word $u$ can be decomposed in a unique way $u = f_1(v) \alpha x$ with $v$, $x$ words such that $|x|_\alpha = 0$.

As $u$ is left special, there exist distinct letters $\beta$ and $\gamma$ 
such that $\beta u$ and $\gamma u$ are factors of $\bw$.
As the letter $\alpha$ marks the beginning of images of letters in $\bw$
and as for all letters $\delta$, $f_1(\delta)$ ends with $\delta$,
we deduce that 
$\beta v $ and $\gamma v$ are factors of $\bw_2$.
As $|v| < |u|$ and by choice of the triple $(\mbf, \bw, u)$,  the word $v$ is a prefix of $\bw_2$.
Consequently $f_1(v)\alpha$ is a prefix of $\bw$ and so $x \neq \varepsilon$.

Assume there exists a unique letter $b$ such that 
$\beta vb$ is a factor of $\bw$ and $u$ is a prefix of $f(vb)$.
Assume also that $b$ is the unique letter $c$ such that $\gamma vc$ 
is a factor of $\bw$ and $u$ is a prefix of $f(vc)$.
As $u$ is not a prefix of $\bw = f_1(\bw_2)$ and as $u$ is a prefix of $f_1(vb)$,
the word $vb$ is not a prefix of $\bw_2$.
By choice of the triple $(\mbf, \bw, u)$, $|vb| \geq |u|$.
As $|v|<|u|$, we get $|vb| = |u| = |f_1(v)\alpha x|$.
As $|f_1(v)| \geq |v|$, it follows $x = \varepsilon$: a contradiction.

From what precedes, we deduce the existence of two distinct letters
$b$ and $c$ such that $\beta v b$ and $\gamma v c$ are factors of $\bw_2$ 
with $u$ a prefix of $f_1(vb)$ and $f_1(vc)$. As $u$ is not a prefix of $\bw = f_1(\bw_2)$, 
the letter $a$ that follows the prefix $v$ of $\bw_2$ is different from $b$ and $c$.
Hence the word $\bw_2$ is $(a, b, c, \beta, \gamma)$-fragile and $v$ is such a fragility.

\medskip

\textit{Step 2: $f_1$ is LSP $(a, b, c)$-breaking}

By definition of letters $b$ and $c$ at Step~1, the word $\alpha x$ is a common prefix of $f_1(b)$ and $f_1(c)$. Also as $u = f_1(v)\alpha x$ is not a prefix of $\bw$ while $f_1(v)a$ is a prefix of $\bw$,
the word $\alpha x$ is not a prefix of $f_1(a)$.
By Corollary~\ref{carac LSP breaking morphisms},
$f_1$ is $(a, b, c)$-breaking.

\medskip

\textit{Step 3: origin of fragilities of $\bw_2$}

Applying iteratively Lemma~\ref{origine fragilities},
we deduce the existence of an integer $n \geq 2$, a sequence of triples of pairwise distinct letters
$(a_i, b_i, c_i)_{i \in \{2, \cdots, n\}}$, a sequence of $(a_i, b_i, c_i, \beta, \gamma)$-fragilities $(v_i)_{i \in \{2, \cdots, n\}}$ such that:
\begin{itemize}
\item $v_i$ occurs in  $\bw_i$ for all $i \in \{2, \cdots, n\}$;
\item $(a_2, b_2, c_2) = (a, b, c)$ and $v_2 = v$;
\item $|v_{i+1}| < |v_i|$ for all $i \in \{2, \cdots, n-1\}$;
\item $v_n = \varepsilon$.
\item words $v_i a_i$, $\beta v_i b_i$, $\gamma v_i c_i$ are respectively prefixes of the words
$f_i( v_{i+1} a_{i+1})\alpha_i$,
$\beta f_i( v_{i+1} b_{i+1})\alpha_i$,
$\gamma f_i( v_{i+1} c_{i+1})\alpha_i$
where $\alpha_i = \first(f_i)$ for $i \in \{2, \cdots, n-1\}$;
\item $a_n = \first(f_n)$;
\item $\beta b_n$, $\gamma c_n$ belong to ${\rm Fact}(f_n(\alp(\bw_{n+1})))$.
\end{itemize}

\medskip

\textit{Step 4: Conclusion}
Let $(q_i)_{i \geq 1}$ be the sequence of states along a path recognizing $\mbf$: for all $n \geq 1$, $(q_n, f_n, q_{n+1})$ is a transition of ${\cal A}_\bLSP$.

At the end of Step 3, we learn that there exists an $(a_n, b_n, c_n)$-fragility in $f_n(\bw_{n+1})$. 
Hence
$a_n$, $b_n$, $c_n$ are pairwise distinct letters. 
Especially as $a = \first(f_n) \not\in \{b_n, c_n\}$ by properties of bLSP morphisms,
the words $\beta b_n$ and $\gamma c_n$ are factors of images of letters, say $b_n'$ and $c_n'$.
As $\alp(q_n) = \alp(f_n(q_{n+1}))$, 
this implies that $b_n'$ and $c_n'$ belong to $\alp(q_{n+1})$ and $a_n$, $b_n$ and $c_n$ belong to $\alp(q_n)$. Moreover, as $\beta b_n$, $\gamma c_n$ are factors of words in $f_n(\alp(q_{n+1}))$, we deduce that $(a_n, b_n, c_n, \beta, \gamma) \in \set(q_n)$.

By backward induction, we can show that for all $i$, $2 \leq i \leq n$, 
$(a_i, b_i, c_i, \beta, \gamma) \in \set(q_n)$.
Especially $(a_2, b_2, c_2, \beta, \gamma) \in \set(q_2)$.
As, by Step~2, $f_1$ is LSP $(a_2, b_2, c_2)$-breaking and $(q_1, f_1, q_2)$ is a transition of ${\cal A}_\bLSP$, we get our final contradiction.
\qed\end{proof}

\section{\label{sec:conclusion}Conclusion}

Recall that G.~Fici~\cite{Fici2011TCS} asked for a characterization of both finite and infinite words.
Observe that it can be proved that any non-empty finite LSP word $w$ is right extendable to a longer LSP word (That is there exists a letter $a$ occurring in $w$ such that $wa$ is a LSP).
As a consequence one can prove:

\begin{lemma}
\label{L:finite}
A finite word is LSP if and only if it is a prefix of an infinite LSP word.
\end{lemma}

Lemma~\ref{L:finite} is a consequence of next result.

\begin{lemma}
\label{L:extendable}
Any non-empty finite LSP word $w$ is right extendable to a longer LSP word. That is there exists a letter $a$ occurring in $w$ such that $wa$ is a LSP.
\end{lemma}

\begin{proof}
The result is immediate for words of length 1.
Let $w$ be a finite LSP word of length at least two.
Let $u$ be its longest border (that is the longest word distinct from $w$ that is both a prefix and a suffix of $w$) and let $a$ be the letter such that $ua$ is a prefix of $w$.
From now on we prove that $wa$ is LSP.
Consider a left special factor of $wa$. If it is a factor of $w$, as $w$ is LSP, it is a prefix of $w$ and so of $wa$.
Assume now that $va$ is both a suffix and a left special factor of $wa$.
As $v$ is then a left special factor of $w$, $v$ is a prefix of $w$ and so it is a border of $w$. 
If $v = u$ then, by definition, $va$ is a prefix of $wa$.
In the remaining case $v \neq u$, $va$ is a suffix of $ua$ which is a prefix of $w$. Hence $va$ occurs in $w$, and as a left special factor, it is a prefix of $w$ and so of $wa$. Hence $wa$ is LSP.
\qed\end{proof}

\noindent
\textit{Proof of Lemma~\ref{L:finite}.}
One can observe that if $u$ is the longest border
of a word $w$ and if $ua$ is a prefix of $w$, then $ua$ is the longest border of $wa$. Hence iterating the proof of the previous lemma, one obtains that:
if $w = pu$ with $u$ the longest border of $w$, then the infinite extension $p^\omega$ of $w$ is LSP. As any prefix of an LSP word is also LSP, Lemma~\ref{L:finite} follows.

By definition, left special factors are prefixes of LSP words. Hence readers can verify that
Lemma~\ref{L:extendable} can not be stated for left extendability. More precisely, if, for a letter $a$ and a word $x$, both $a$ and $x$ are LSP, then $ax$ has no left special factor except the empty word or $ax = a^nu$ for a word $u$ whose left special factors are words $a^i$ with $i \leq n$.\qed

\medskip

Lemma~\ref{L:finite} shows that any characterization of infinite LSP words provides naturally a characterization of finite LSP words (adding ``is a prefix of" before the characterization of infinite LSP words). For instance in the binary case, this allows to find back M.~Sciortino and L.Q.~Zamboni's result \cite{Sciortino_Zamboni2007DLT}: ``binary words having suffix automaton with the minimal possible numbers of states are exactly the finite prefixes of standard Sturmian words" (that can be reformulated after G.~Fici's work : ``finite binary LSP words are exactly the finite prefixes of standard Sturmian words").
For this purpose, one can first see from Theorem~\ref{th_carac} and Figure~\ref{Graph2} that directive words of binary infinite LSP words are ultimately $\tau_a$ or ultimately $\tau_b$ or 
ultimately contain both $\tau_a$ and $\tau_b$. By classical results (see, \textit{e.g.}, \cite{BerstelSeebold2002Lot}) it can be deduced that an infinite LSP word is an infinite repetition of a finite standard word or is an infinite standard word. As any power of a finite standard word is a prefix of an infinite standard word (see \cite[Chap. 2]{Lothaire2002} for instance), we get M.~Sciortino and L.Q.~Zamboni's result.

We end this paper mentioning natural questions arising from this work.
Can a smaller automaton than ${\cal A}_\bLSP$ can be found for recognizing directive words of LSP infinite words?
Can a similar $S$-adicity system can be found for infinite words having at most one left special factor? 
Does there exist a finite or infinite set $S$ of morphisms such that an infinite word is LSP if and only if it $S$-adic (as it occurs for infinite balanced binary words)?

\end{document}